\documentclass[aps,prx,superscriptaddress,floatfix,10pt,twocolumn]{revtex4-2}
\pdfoutput=1
\usepackage[utf8]{inputenc} 

\usepackage{newtxtext} 
\usepackage{amssymb,amsfonts,mathtools}
\usepackage{bm,dsfont,physics}
\usepackage{amsthm}
\usepackage{nicefrac}

\usepackage{multirow}

\usepackage{quantikz}

\usepackage{graphicx}
\usepackage[caption=false]{subfig} 
\usepackage{overpic}
\usepackage{float}
\usepackage{color}
\usepackage{multirow}
\usepackage{booktabs}
\usepackage{rotating}

\usepackage{latexsym,verbatim,outlines}
\usepackage[normalem]{ulem}
\usepackage{lipsum}

\usepackage{amsthm}
    \newtheorem{theorem}{Theorem}

\usepackage[dvipsnames]{xcolor}
\PassOptionsToPackage{dvipsnames}{xcolor}
\usepackage[breaklinks=true,colorlinks,citecolor=Violet,linkcolor=Violet,urlcolor=Violet]{hyperref}

\footnotetext{These authors contributed equally.}

\begin{document}
\title{
Quantum error correction with global control
}

\author{Roberto Menta$^{\S}$}
\email{rmenta@planckian.co}
\affiliation{Planckian srl, I-56127 Pisa, Italy}
\affiliation{NEST, Scuola Normale Superiore, I-56127 Pisa, Italy}
\author{Lindsay Bassman Oftelie$^{\S}$}
\email{loftelie@planckian.co}
\affiliation{Planckian srl, I-56127 Pisa, Italy}
\author{Ashkan Abedi}
\affiliation{Planckian srl, I-56127 Pisa, Italy}
\author{Francesco Cioni}
\affiliation{NEST, Scuola Normale Superiore, I-56127 Pisa, Italy}
\author{Marco Polini}
\affiliation{Planckian srl, I-56127 Pisa, Italy}
%
\author{Seth Lloyd}
\affiliation{Planckian srl, I-56127 Pisa, Italy}
\affiliation{Department of Mechanical Engineering, Massachusetts Institute of Technology, Cambridge, MA-02139, USA}
\author{Francesco Caravelli}
\affiliation{Planckian srl, I-56127 Pisa, Italy}
\author{Vittorio Giovannetti}
\affiliation{Planckian srl, I-56127 Pisa, Italy}

\begin{abstract}
Reaching fault tolerance means scaling qubit counts by orders of magnitude, a jump that conventional superconducting architectures cannot sustain without solving the so-called `wiring problem'.
Global control sidesteps this bottleneck, but implementing quantum error correction (QEC) on previously proposed global architectures incurs extremely steep overhead costs, due to the need for separate correction procedures for the computational and auxiliary qubits that comprise the global device. 
We resolve this by introducing the first globally-controlled architecture with zero qubit overhead.  Every physical qubit is a computational qubit, and thus, every qubit is protected under a single error correcting scheme. We identify a class of cyclic stabilizer codes realizable through global iSWAP and single-qubit gates, yielding QEC thresholds nearly seven orders of magnitude larger than previous estimates for globally-controlled arrays. We further show these thresholds improve systematically as the global architecture is augmented with a limited amount of local measurement sites, demonstrating a trade-off between wiring simplicity and fault-tolerant performance.
\end{abstract}

\maketitle

Quantum computing has emerged as a powerful paradigm for solving a range of problems that are intractable on classical computers, from quantum simulation to optimization and cryptography~\cite{Huang2025, Babbush2025}. Yet, despite rapid progress, current devices remain limited by noise, finite coherence times, and restricted qubit counts, defining the so-called NISQ era~\cite{Preskill2018}. Thanks to substantial technological advances and sustained global investment, quantum processors are now beginning to move beyond this regime~\cite{Preskill2025, Eisert2025}. As systems scale toward hundreds or thousands of qubits, however, hardware architecture and control complexity become central bottlenecks~\cite{Mohseni2024}.

\begin{figure}[!t]
    \centering
    \includegraphics[width=\linewidth]{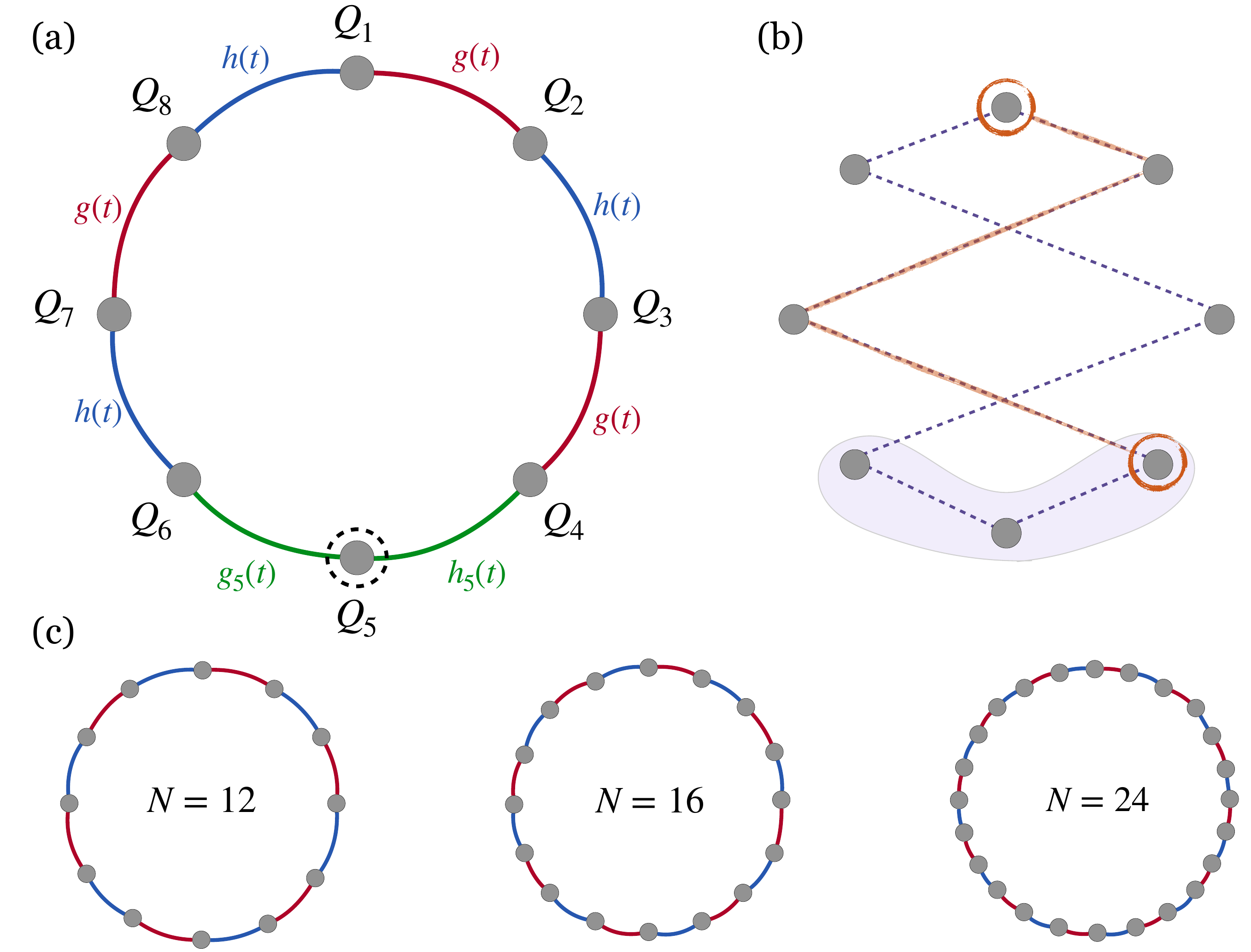}
    \caption{Globally controlled quantum architecture with zero qubit overhead.
    (a) The architecture possesses a ring topology with $N$ (even) computational qubits $\mathcal{Q}_i$
    (gray circles), here $N=8$, arranged on a circle with XY nearest-neighbor couplings.
    Red (blue) links denote the global coupling $g(t)$ [$h(t)$], acting uniformly on the odd (even) bonds; even- and odd-indexed qubits are addressed by two separate global drives. Green links denote the two
    locally controlled couplings $g_{\bar{i}}(t)$ and $h_{\bar{i}}(t)$ on the
    bonds adjacent to the single locally addressable site $\bar{i}$ (dashed circle;
    $\bar{i} = 5$ in the figure), where single-qubit gates are also applied.
    Computational two-qubit gates are performed on these two bonds.
    (b) Effective connectivity graph of the architecture for $N=8$. Vertices are the
    computational qubits; an edge links two qubits that the global transport dynamics can bring simultaneously into the locally controlled region (shaded). The graph is
    connected, which, together with the local single- and two-qubit gates at $\bar{i}$, is sufficient for universal quantum computation. Since every vertex has degree two, an entangling gate
    between qubits not joined by an edge must be compiled as a chain of
    operations along a path of the graph; the orange path shows the shortest such chain for the highlighted pair $Q_1$ and $Q_4$, of length three.
    (c) Larger rings, with $N=12$, $16$, and $24$. The number of independent controls (two global couplings plus local controls at the single addressable site) remains constant in all cases, independent of $N$.}
    \label{fig:ring-int}
\end{figure}

In leading platforms such as superconducting circuits, each qubit typically requires several dedicated control and readout lines~\cite{Kjaergaard2020}. This results in severe wiring congestion, substantial control complexity, increased thermal load at cryogenic stages, intricate calibration protocols, and strong geometric constraints on chip layout~\cite{Marvin2026}. As the number of qubits grows~\cite{Google2023, Google2025}, these challenges worsen superlinearly, giving rise to the so-called `wiring problem' and raising fundamental questions about the long-term scalability of architectures based on fully local addressability.

Global control~\cite{Lloyd1993} has emerged as a compelling alternative to conventional local control. In globally-driven architectures, only a few external fields act simultaneously on the entire qubit register~\cite{Benjamin2000}. Early theoretical work established that universal quantum computation is still attainable under such global operations~\cite{LloydQAOA}, provided that suitable structural asymmetries are present~\cite{Lloyd1993, Benjamin2000, benjamin_2001, Levy_2002, benjamin_2003, benjamin-bose_2004, Fitzsimons_2006, Silva_2009}. More recent studies have demonstrated versions of global control adapted to realistic hardware models~\cite{cesa2023universal, Patomaki2024, menta2024globally, cioni2024conveyorbelt, menta2025building}. 

However, such schemes rely on ingredients such as multiple qubit species and patterned energy splittings ~\cite{menta2026actuators} which present significant fabrication challenges, undermining the simplicity that global control seeks to provide. Worse, all previously presented global control architectures require a fraction of the physical qubits to be devoted to non-computational infrastructure used to mediate addressing without local control.  These auxiliary qubits, which are distinguished from the computational qubits that store the quantum information, present a major challenge for implementing quantum error correction (QEC). While standard QEC schemes can be mapped to global architectures to stabilize the computational qubits, they leave the auxiliary qubits unprotected. To address this, previous works have proposed running two separate, interleaved error correction procedures, one for the computational qubits and one for the auxiliary qubits \cite{Bririd_2004, Kay_2005, Kay_2007,Fitzsimons_2008}. This, however, generates massive temporal and spatial costs as the procedures cannot be performed in parallel (forcing the device to cycle between the two) and each requires its own qubit redundancy. It remains unclear if such methods could even be effective within realistic hardware and timing constraints. 

Here, we present a scalable, global-control scheme that requires no intricate species engineering and, for the first time, \textit{zero} auxiliary qubit overhead.  Our model consists of a homogeneous ring of qubits coupled via nearest-neighbor exchange interactions whose strengths can be globally switched between two configurations ~\cite{benjamin_2001, Wu2004, Ye2006}, see Fig.~\ref{fig:ring-int}. By alternating these global coupling patterns, the system naturally generates a programmable sequence of pairwise exchange operations, enabling quantum information to be transported around the ring purely through global dynamics. Universality is obtained by incorporating a single locally driven qubit, which provides the minimal symmetry breaking necessary to implement arbitrary single-qubit rotations~\cite{Hu2025, Gargiulo2026}. Any computational qubit can be routed to this site via global transport dynamics, where local gates are applied. Two-qubit entangling gates are generated using the same mechanism: information transport brings the relevant qubits into the locally controlled region, where controlled interactions are executed before the qubits are returned to their original locations. 

Using only a single qubit species, one locally controlled element, and four global control channels (two couplings and two drives), our architecture represents the most resource-efficient global-control scheme proposed to date---Fig.~\ref{fig:ring-int}(c) displays the architecture as the number of qubits is varied. Such simplicity makes it directly compatible with superconducting-qubit platforms featuring tunable couplers and flux control~\cite{benjamin_2003, benjamin-bose_2004, Kounalakis_2018, Yan2018, Mitarai2025}, offering a realistic avenue for experimental implementation. Moreover, the lack of auxiliary qubit infrastructure allows error correction to be implemented uniformly across the whole device, yielding a much cleaner realization of globally-controlled fault tolerance. 

We identify a stabilizer code with syndrome extraction circuits that are natively comprised of global gates, which we refer to as the cyclic code.  Initially described for a 5-qubit encoding~\cite{Schuch2003}, the code was recently adapted for execution on a ring topology for higher code distances~\cite{Simakov2025}, making it well-matched for our proposed architecture.  We perform Monte Carlo simulations to estimate QEC thresholds for the cyclic code on our architecture.  The dramatic simplification of error correction in our architecture is exemplified by the thresholds of $\sim 10^{-3}$ we obtain, nearly seven orders of magnitude larger than thresholds previously derived for dual computational-auxiliary global control schemes~\cite{Kay_2007}. While our minimal model features a single local measurement site, we demonstrate how distributing additional measurement sites uniformly around the ring can improve QEC performance, exposing an interesting trade-off between wiring simplicity and fault-tolerant performance~\cite{Menta2026ETH}.

\textit{The model.}---Consider an even number $N$ of  interacting qubits $\mathcal{Q}_1, 
\cdots, \mathcal{Q}_{N}$ arranged in a closed ring as shown in Fig.~\ref{fig:ring-int}(a) for the case of $N = 8$.  The Hamiltonian (setting $\hbar = 1$) is  
\begin{equation}
    \label{eq:H-tot1}
    \hat{H}(t) := \hat{H}_0 + \hat{H}_{{\rm drive}}(t) + \hat{H}_{{\rm int}}(t) \ ,
\end{equation}
where $\hat{H}_0:= \sum_{i =1}^N \frac{\omega_{\mathcal{Q}}}{2} \hat{\sigma}_i^{(z)}$
represents the local free energies of qubits, which we assume to  share the same frequency $\omega_{\mathcal{Q}}$, and $\hat{\sigma}^{(x,y,z)}_i$ are the Pauli operators acting on the $i$th site.
The second term in Eq.~(\ref{eq:H-tot1}) represents a driving contribution, which consists of local terms activated by dedicated global control lines. It comprises three components: 
$\hat{H}_{\rm drive}(t) := \hat{H}_{{\rm drive},{\rm E}}(t) + \hat{H}_{{\rm drive},{\rm O}}(t) + \hat{H}_{{\rm drive},\bar{i}}(t)$. The first two
\begin{align}
    \label{eq:odd-even}
    \hat{H}_{{\rm drive},{\rm E/O}}(t) := \Gamma_{\rm E/O}(t)  \sum_{i\in {\rm E/O}} \hat{\sigma}_{{i}}^{(y)}
\end{align}
act on the even- (E) and odd-indexed (O) qubit subsets, respectively, through two independent global pulses $\Gamma_{\rm E/O}(t) := \Omega_{{\rm E/O}}(t)\sin(\omega_{\mathrm{d}}t + \phi_{{\rm E/O}}(t))$. The third component addresses a single qubit of the ring, $\mathcal{Q}_{\bar{i}}$, ($\bar{i}=5$ in Fig.~\ref{fig:ring-int}(a)), which is one of the few elements of the model requiring local control,
\begin{align}
    \hat{H}_{{\rm drive},\bar{i}}(t) := \Gamma_{\bar{i}}(t)  \hat{\sigma}_{\bar{i}}^{(y)}
\end{align}
through the pulse $\Gamma_{\bar{i}}(t) := \Omega_{\bar{i}}(t)\sin(\omega_{\mathrm{d}}t + \phi_{\bar{i}}(t))$.

The last term in Eq.~(\ref{eq:H-tot1}) describes nearest neighboring  Heisenberg XY interactions $\hat{V}_{i,i+1}:= \frac{1}{2}(\hat{\sigma}^{(x)}_{i} \hat{\sigma}^{(x)}_{i+1} + \hat{\sigma}^{(y)}_{i} \hat{\sigma}^{(y)}_{i+1})$. Specifically, we assume  
\begin{equation}
\hat{H}_{{\rm int}}(t):= \hat{H}_{{\rm int},\mathcal{P}_{\rm E}}(t)  + \hat{H}_{{\rm int},\mathcal{P}_{\rm O}}(t) + \hat{H}_{{\rm int},{\bar i}^{+}}(t)+\hat{H}_{{\rm int},{\bar i}^{-}}(t)\ .
\end{equation}
Here, $\hat{H}_{{\rm int},{\bar i}^{+}}(t) := h_{\bar i}(t)  \hat{V}_{{\bar i},{\bar i}+1}$ and $\hat{H}_{{\rm int},{\bar i}^{-}}(t) := g_{\bar i}(t)  \hat{V}_{{\bar i}-1,{\bar i}}$ are the interactions between $\mathcal{Q}_{\bar i}$ and its neighboring sites, $\mathcal{Q}_{{\bar i}+1}$ and $\mathcal{Q}_{{\bar i}-1}$,  mediated by the time-dependent, interaction strengths $h_{\bar i}(t)$ and $g_{\bar i}(t)$ which can be independently controlled. Meanwhile, $\hat{H}_{{\rm int},{\mathcal{P}_{\rm E}}}(t) = h(t) \sum_{i\in{\rm E}} \hat{V}_{i,i+1}$  and $\hat{H}_{{\rm int},{\mathcal{P}_{\rm O}}}(t) = g(t) \sum_{i\in{\rm O}} \hat{V}_{i,i+1}$ describe uniform couplings between the elements of the pairs subsets $\mathcal{P}_{\rm E}:= \{(\mathcal{Q}_i,\mathcal{Q}_{i+1})\}_{ i\in {\rm E}}$ and $\mathcal{P}_{\rm O}:= \{(\mathcal{Q}_i,\mathcal{Q}_{i+1})\}_{ i\in {\rm O}}$, respectively. The associated uniform coupling strengths $h(t)$ and $g(t)$ are independent from $h_{\bar i}(t)$ and $g_{\bar i}(t)$, but like the latter are also controlled externally. In particular, they allow us to  selectively turn off the interactions within the subsets $\mathcal{P}_{\rm E}$ and $\mathcal{P}_{\rm O}$  at different times, transforming the system into a collection of isolated pairs. We refer to these two drive windows, $\mathcal{T}_{h/g}$, as $h$-drive and $g$-drive, respectively, i.e.,
\begin{eqnarray}
        &&\mathcal{T}_h : \ g(t) = g_{\bar{i}}(t) = 0 \  \wedge \ h(t) = h_{\bar{i}}(t) = J  \ ,  \label{eq:alternative-g-h} \\
        &&\mathcal{T}_g : \ h(t) = h_{\bar{i}}(t) = 0 \  \wedge \ g(t) = g_{\bar{i}}(t) = J \ , \label{eq:alternative-h-g}
\end{eqnarray}
where $J$ is an engineered target coupling strength. 
By alternating between these two drives the system evolves under a non-trivial dynamics which we make explicit below.

\textit{Information transport.}---At its core, our scheme uses global control fields to route information to a local-control zone where single- and two-qubit gates are performed, thereby achieving universal computation. This routing-based picture offers a general framework for understanding global control. While all previously proposed global-control schemes require auxiliary qubits (e.g., buffer qubits, pointer qubits, etc.)~\cite{Lloyd1993, Benjamin2000, benjamin_2003, Fitzsimons_2006, cesa2023universal, menta2024globally, cioni2024conveyorbelt, menta2025building} for implementing transport of the computational (i.e., information-storing) qubits, we emphasize that in our model, all $N$ qubits are computational qubits, resulting in zero qubit overhead for global control.

Realizing this routing framework concretely requires a mechanism that allows each individual qubit state to be transferred through the register using only the global fields. To this end, we engineer a compilation scheme that transports the qubit states around the ring through exchange operations between neighboring qubits of different parity. We begin by noting that the ability to alternately ``turn on'' and ``turn off'' the coupling strengths, as shown in Eqs.~\eqref{eq:alternative-g-h} and \eqref{eq:alternative-h-g}, implies that the temporal evolution can be divided into a sequence of alternating disjoint time windows. By setting the interaction time $\tau_{h/g} = \frac{3\pi}{2J}$, the system executes an iSWAP gate. However, for seamless transport of quantum information, a conventional SWAP gate is required. 

In our framework, we can achieve a pure SWAP, e.g., by decomposing it into a sequence of iSWAP gates and global rotations. Specifically, as shown in Ref.~\cite{Schuch2003, Marvian2024} and in the Supplemental Material (SM)~\cite{SM}, this can be done using three iSWAP operations and global $\sqrt{X} \simeq R_X(\frac{\pi}{2})$ pulses:
\begin{equation}
    \resizebox{\columnwidth}{!}{$
    \text{SWAP} =  (I \otimes \sqrt{X}) \ \text{iSWAP} \ (\sqrt{X} \otimes I) \ \text{iSWAP} \  (I \otimes \sqrt{X}) \  \text{iSWAP},
    $}
\end{equation}
where each component of the tensor products acts globally either on the even-parity qubits or the odd-parity one. If the hardware supports tunable couplers capable of mediating both XY and ZZ interactions~\cite{Sung2021}, one can achieve the same transformation more compactly~\cite{krizan2025} using the coupler to perform both a CZ and an iSWAP:
\begin{equation}\label{iswap_to_swap_local}
    \text{SWAP} = (S^\dagger \otimes S^\dagger) \, \text{iSWAP} \, \text{CZ} \, ,
\end{equation}
where $S^{\dagger} \simeq R_Z(-\pi/2)$. Of course, one could alternatively use optimal control to design a sequence of pulses that directly implements a SWAP operation \cite{Huang2014OptimalControl, Spiteri2018QuantumControl, Cho2024DirectPulse}.  By interleaving SWAP blocks, the dynamics can selectively move the internal state of any qubit $\mathcal{Q}_i$ into the locally driven site ${\bar{i}}$~\cite{SM} where arbitrary local operations can be realized.

\textit{Universal quantum computation.}---To achieve universality in a globally-controlled quantum processor, the system Hamiltonian must incorporate a ``global symmetry-breaking'' mechanism~\cite{Hu2025, Gargiulo2026}. In our architecture, this role is fulfilled by the locally controlled terms in Eq.~\eqref{eq:H-tot1}: the single-qubit drive $\hat{H}_{{\rm drive},\bar{i}}(t)$ and the local interaction terms $\hat{H}_{{\rm int},\bar{i}^{\pm}}(t)$. The single-qubit term implements local rotations on the qubit currently residing at site $\bar{i}$.

By applying arbitrary iterations of the SWAP transport protocol, the logical state of any qubit in the register can be routed to the $\bar{i}$-th site. Once positioned there, we execute the local gate protocol:
\begin{equation}
    \mathcal{T}_{\rm 1q}: g(t)= h(t) = g_{\bar{i}}(t) = h_{\bar{i}}(t)= 0 \ \wedge \ \Omega(t), \phi(t) \neq 0.
\end{equation}
Setting the drive frequency to resonance ($\omega_{\rm d} = \omega_{\mathcal{Q}}$), we implement an arbitrary single-qubit rotation $\hat{U}_{{\rm 1q},\bar{i}} = \exp\!\left[-\frac{i\theta}{2}\, \mathbf{n}_\phi \cdot \bm{\sigma}_{\bar{i}}\right]$ (up to a global phase). Here, $\bm{\sigma}_{\bar{i}}$ is the Pauli vector, $\mathbf{n}_\phi = (\cos\phi, \sin\phi, 0)$ is the rotation axis in the equatorial plane, and the rotation angle is $\theta = \int_0^{\tau_{\rm 1q}} \Omega(t')\, \text{d}t'$.

Beyond single-qubit rotations, computational universality requires entangling gates between qubits. This is achieved via the two locally controlled XY interactions with strengths $h_{\bar{i}}(t)$ and $g_{\bar{i}}(t)$. The entangling protocol is defined by:
\begin{equation}
     \mathcal{T}_{\rm 2q}: \ g(t)= h(t) = 0 \quad \wedge \quad 
     \begin{cases}
         g_{\bar{i}}(t) = 0, \quad h_{\bar{i}}(t) = \tilde{J} \\
         g_{\bar{i}}(t) = \tilde{J}, \quad h_{\bar{i}}(t) = 0 
     \end{cases},
\end{equation}
where $\tilde{J}$ is the engineered coupling strength. Evolution over a time window of duration $\tau_{\rm 2q} =\frac{3\pi}{2\tilde{J}}$ generates an iSWAP gate between the pairs $(\mathcal{Q}_{\bar{i}-1}, \mathcal{Q}_{\bar{i}})$ or $(\mathcal{Q}_{\bar{i}}, \mathcal{Q}_{\bar{i}+1})$. As arbitrary single-qubit gates combined with iSWAP constitute a universal set~\cite{Nielsen2010}, our architecture is formally universal (see Fig.~\ref{fig:ring-int}(b) and the End Matter).

A realization of this architecture with superconducting qubits, consisting of always-on XY interactions and global control on the internal frequencies of the qubits, is presented in the SM~\cite{SM}. 

\textit{Quantum error correction.}---Standard QEC implementation requires the computation and measurement of syndromes, which together we term syndrome extraction (SE)~\cite{Nielsen2010}. Circuits for performing SE first entangle data qubits (which store the logical quantum state) with syndrome qubits (ancillary qubits used to read out the syndromes), followed by measurement of all syndrome qubits.  Depending on the spatial distribution of syndromes in the architecture and the connectivity of the device, SE circuits may contain operations that act in parallel across the register (and can thus be implemented by global gates) as well as local operations addressing single qubits. In our global architecture, executing local gates on different qubits requires moving the target qubit to the local-control site, applying the gate, and moving the next target qubit into the local-control zone. As the number of physical qubits increases, such serial execution will inhibit the performance of QEC, as errors can start to accumulate more quickly than they can be corrected.  Therefore, care must be taken to select a QEC code compatible with the global architecture.

To this end, we identify a cyclic stabilizer code $[\![n,1,d]\!]$ encoding one logical qubit into $n$ data qubits, in which syndrome computation is implemented \textit{solely} with global gates, eliminating the need for routing qubits to the local-control zone \cite{Simakov2025}.  Furthermore, the code possesses a circular topology perfectly matched to our ring architecture, allowing for a straightforward transpilation of the SE circuits into a series of global iSWAP gates and global single-qubit rotations. The only part of the SE procedure that must proceed serially is measurement of the syndromes, as global control presumes not to have local read-out on every site.  Measurement of syndromes is thus implemented with alternating rounds of (i) measuring the qubit(s) at the local measurement site(s) and (ii) performing global SWAPs to carry the next set of syndrome qubits to the measurement zone(s). This process must be repeated until all syndromes have been measured.  In our minimal model, the ring is endowed with one local measurement site, but as we will show, distributing additional local measurement sites around the ring can improve QEC performance. The number of measurement rounds that must be executed is given by $M = \lceil \frac{n_s}{\ell} \rceil$, where $n_s = n - 1 = 4d - 8$ is the number of syndromes to be measured for the cyclic code with distance $d$, and $\ell$ is the number of measurement sites. The total ring will contain $N=2n$ qubits, with half designated as data qubits and half designated as syndrome measurement qubits. See the End Matter for more details on the cyclic code and its global implementation.

We perform Monte Carlo simulations of the cyclic code on our architecture to estimate logical error rates, following the method presented in Ref.~\cite{Simakov2025}. We generate circuits with a varying number of back-to-back SE rounds, between 1 and 50, and simulate them in a noisy environment with the Stim package~\cite{gidney2021stim}. We assume a phenomenological noise model based on the two major sources of error. The first is depolarizing error~\cite{Nielsen2010} on the data qubits while they remain idle during syndrome qubit measurement. We account for this by inserting a depolarizing channel on the data qubits in between SE rounds.  The second is measurement error, which we implement with a bit-flip error channel on the syndrome qubits right before they are measured.  Finally, for decoding we use a look-up table augmented with a `memory' \cite{Simakov2025}, which enables the correction of both spacelike and timelike errors.

\begin{figure}[!t]
    \centering
    \begin{overpic}[width=0.94\columnwidth]{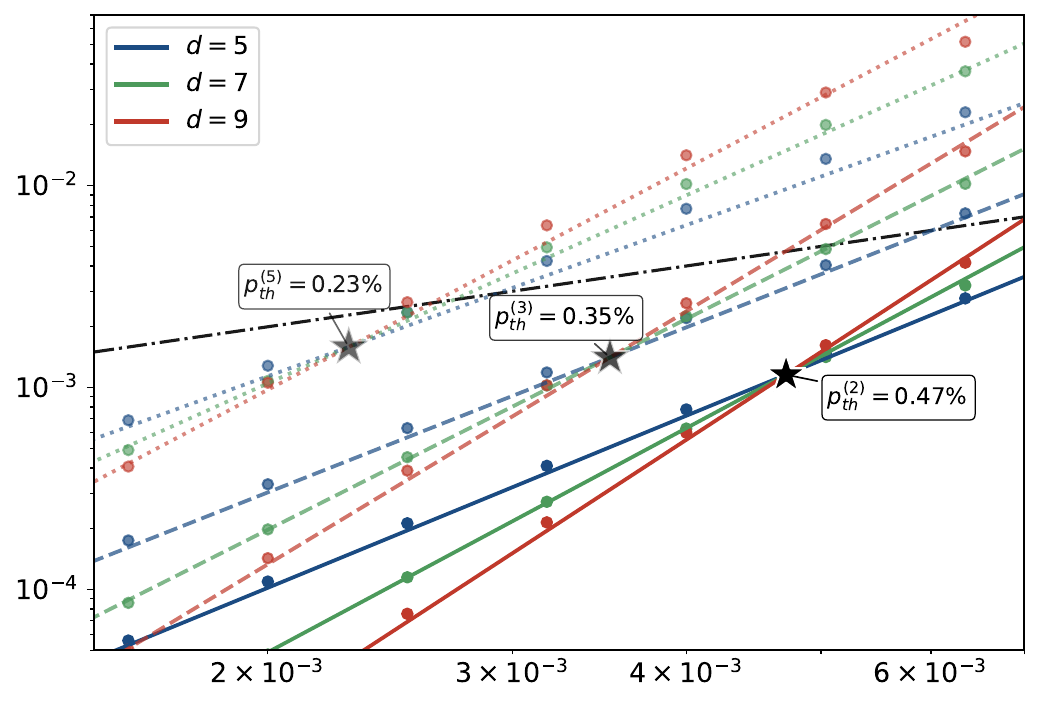}
    \put(-3,20){\rotatebox{90}{Logical Error Rate}}
    \put(38,-3){Physical Error Rate}
    \end{overpic}
    \caption{Logical error rate versus physical error rate for the cyclic code executed on the globally-controlled ring architecture, for code distances $d=5$ (blue), $d=7$ (green), and $d=9$ (red), using enough measurement sites for $M=5$ (dotted curves), $M=3$ (dashed curves), and $M=2$ (solid curves) rounds of measurement (see Table~\ref{tab:wires}). Circles represent data obtained from Monte Carlo simulations under the phenomenological noise model, while lines are fit to guide the eye.  The mutual crossing point within each set of code distances gives the QEC threshold $p^{(M)}_{\rm th}$, demarcated with a star. The dashed-dotted black line marks break-even (logical error rate equal to physical error rate).} 
    \label{fig:threshold}
\end{figure}

While a single measurement site suffices in principle, achieving strong QEC performance will likely require additional sites to keep syndrome-extraction SE circuits at a reasonable depth.  Figure~\ref{fig:threshold} plots the logical error rate as function of the physical error rate for code distances $d=5$ (blue), $d=7$ (green), and $d=9$ (red), using enough measurement sites for each distance to perform complete SE with $M=5$ (dotted curves), $M=3$ (dashed curves), and $M=2$ (solid curves) rounds of measurement.  Table~\ref{tab:wires} shows the number of measurement sites $\ell$ used for each of the distances $d$ and number of measurement rounds $M$. The right-most column provides the approximate percentage of all qubits (data and syndromes) endowed with local measurement for a given $M$, giving a rough idea of the measurement wiring overhead for each instance. For example, the distance $d=5 (d=7)$ code has 13(21) stabilizer generators to be measured (see End Matter), so $M=5$ rounds of measurement would require $\ell=3 (\ell=5)$ local measurement sites, which for all code distances equates to approximately 11$\%$ of all qubits featuring local read-out. 

For fixed $M$, we identify the QEC threshold $p^{(M)}_{\rm th}$ as the intersection point of the curves for the three code distances, demarcated with a star.  Note that fixing $M$ means supplying a different number of local measurement sites $\ell$ for each code distance $d$.   We observe that QEC thresholds of $\sim 10^{-3}$ can be obtained with low measurement wiring overhead.  Previous estimates for globally-controlled arrays found a threshold on the order of $10^{-10}$ when taking into account the limited parallelism of gate execution and measurement due to global control~\cite{Kay_2007}, as we likewise do here. It should be noted that little was done in that work to optimize the threshold, as the goal was simply to evidence the \textit{existence} of a threshold, but this alone cannot account for the nearly seven order-of-magnitude improvement we exhibit here. Furthermore, we observe that the threshold can be systematically improved by increasing the number of local measurement sites (i.e., lowering $M$). This provides a path for obtaining a desired level of QEC performance by tuning the number of local measurement sites, subject to the practical constraints imposed by achievable measurement wiring density.

\renewcommand{\arraystretch}{1.3}
\begin{table}[htbp]
    \centering
    
    
    \begin{tabular}{c c c c c @{\hspace{\tabcolsep}\vrule width 1.2pt\hspace{\tabcolsep}} c}
    
        & & \multicolumn{3}{c}{\textbf{Code distance ($d$)}} & \\
        \cmidrule(l){3-5}
        

        & & \makebox[1.3cm]{\textbf{5}} & \makebox[1.3cm]{\textbf{7}} & \makebox[1.3cm]{\textbf{9}} & \makebox[1.3cm]{\textbf{($\%$)}} \\
        
        \midrule
        \multirow{4}{*}{\rotatebox{90}{\textbf{Rounds ($M$)}}}
        & \textbf{2} & 7 & 11 & 15 & $\sim$ 26 \\
        & \textbf{3} & 5 & 7 & 10 & $\sim$ 18 \\
        & \textbf{4} & 4 & 6 & 8 & $\sim$ 14 \\
        & \textbf{5} & 3 & 5 & 6 & $\sim$ 11 \\
        \bottomrule
    \end{tabular}
    \caption{Number of measurement sites required for each code distance $d$ to measure the complete set of syndromes within $M$ rounds of measurement. The right-most column gives the approximate percentage of all qubits (data plus syndrome) that are endowed with local measurement.}
    \label{tab:wires}
\end{table}

\textit{Discussion.}---We introduced the most resource-efficient globally-controlled architecture for quantum computation to date. Requiring only a single qubit species, one locally controlled element, and four global control channels, our scheme maintains full computational universality.  A key advantage is that such universality is obtained without species engineering, patterned couplings, or redundant encoding: each computational qubit corresponds to a single physical two-level system, with no auxiliary qubits required. This eliminates the dominant complexity present in earlier global-control schemes~\cite{Lloyd1993, Benjamin2000, benjamin_2001, benjamin_2003, benjamin-bose_2004, Fitzsimons_2006}. A lack of auxiliary qubits also enables a clean implementation of QEC, as a single error-correcting scheme can be uniformly applied to protect all qubits in the register.

We also outlined a realistic implementation with superconducting qubits.  Unlike previous globally driven superconducting proposals~\cite{menta2024globally, cioni2024conveyorbelt, menta2025building}, whose reliance on large always-on ZZ interactions~\cite{Riccardi2026} comes at the cost of some physical overhead, our architecture exploits fast XY (iSWAP) dynamics, making it computationally comparable to standard locally controlled processors.  Although spatial inhomogeneities remain an experimental challenge, recent work~\cite{aiudi2025} shows that globally optimized pulses can compensate disorder, indicating that our approach is compatible with realistic superconducting hardware.  Such a simplification allows us to achieve QEC thresholds almost seven orders of magnitude greater than previous estimates for globally-controlled arrays.

While the cyclic code is advantageous in terms of qubit count (linear as opposed to quadratic scaling with code distance) and its adaptability to global control, the drawback is that its syndrome weights grow with increasing code distance, excluding it from the class of quantum low-density parity check (qLDPC) codes~\cite{LDPC1, LDPC2, Panteleev2022, Breuckmann2021}, which show the most promise for scaling quantum computers up to fault tolerance~\cite{gu2026nearest, nixon2026vine,
geher2025directional}.  
Developing qLDPC-based, globally driven fault-tolerant schemes tailored to our architecture is therefore an exciting avenue for future work.  

To this end, we highlight the fact that the effective connectivity of our architecture mirrors that of a 2D lattice and can even approach all-to-all connectivity with ``keyhole'' geometries~\cite{SM}. Such enhanced connectivity may ease implementation of qLDPC codes, such as bicycle and hypergraph-product constructions~\cite{Tillich2014, Kovalev2013, Hastings2021} by removing geometric constraints which generally limit their implementation on standard planar architectures.  

Beyond encoding a single logical qubit, the ring could act as a building block for a scalable \textit{logical} architecture: multiple rings could be tiled in a two-dimensional network of interconnected rings, with each encoding a logical qubit and coupled to its neighboring rings through a shared physical qubit. Using the global control structure already in place, the physical transport-and-gate mechanism can be lifted to the logical level, using the shared qubit for entangling adjacent rings.  Logical SWAPs between adjacent rings (which can be made transversal for identically encoded rings, hence, making them code-independent) could likewise be implemented through the bridging qubit, thereby providing logical-level routing across the array.

Overall, our proposed architecture provides a realistic, scalable, and experimentally compatible route toward large-scale global-control quantum processors with minimal hardware overhead. Naturally, our scheme can be implemented on any quantum computing platform where iSWAP or SWAP gates are available ~\cite{Bergonzoni2025,Ildefonso2025,Zhu2025, Kiefer2026, kiefer2026digital, calliari2026programmable}. We finally stress that this type of architecture can be further generalized ~\cite{Caravelli-theorem}.

\textit{Acknowledgments.}---The authors thank the Architecture and Hardware Planckian teams for insightful discussions. We thank R. Aiudi for comments on early versions of this scheme.
At the time of their contributions, authors affiliated with Planckian are either employees of Planckian or research collaborators with Planckian.
We acknowledge that this work is fully human-generated. 

\bibliography{biblio}


\clearpage
\setcounter{section}{0}
\setcounter{equation}{0}%
\setcounter{figure}{0}%
\setcounter{table}{0}%

\renewcommand{\theequation}{E\arabic{equation}}
\renewcommand{\thefigure}{E\arabic{figure}}
\renewcommand{\thetable}{E\arabic{table}}
\section*{End Matter}

\subsection*{Universality of the globally controlled ring}

Universality of the scheme follows from a simple property of the \textit{connectivity graph} of the architecture, shown in Fig.~\ref{fig:ring-int}(b) of the main text: its vertices are the computational qubits, and an edge joins any two qubits that the global transport dynamics can bring together into the locally controlled region. Because this graph is \textit{connected}, the global SWAP dynamics can route any qubit of the ring to the local site $\bar{i}$, and any chosen pair of qubits to the locally controlled link. There, two ingredients are available: (i) arbitrary single-qubit rotations, generated by the local drive $\hat{H}_{{\rm drive},\bar{i}}(t)$, and (ii) a two-qubit entangling gate: the iSWAP produced by the local couplings $h_{\bar{i}}(t)$ and $g_{\bar{i}}(t)$. Since arbitrary single-qubit gates together with the iSWAP form a universal gate set~\cite{Nielsen2010}, and the connectedness of the graph guarantees that this set can be enacted on \textit{any} qubit and between \textit{any} pair, global-control transport supplemented by local gates renders the architecture universal~\cite{Caravelli-theorem}.

The same construction can be realized in a complementary hardware setting in which the XY couplings are always on and control is applied globally to the qubit internal frequencies rather than to the interactions: selectively detuning the qubit species reproduces the identical pattern of parallel iSWAP operations (see SM~\cite{SM}). Moreover, the connectivity---and hence the cost of compiling entangling gates between distant qubits---can be increased by deforming the ring into a \textit{keyhole} or \textit{lollipop} geometry, in which a short, locally controlled corridor is appended to the ring; at the price of a small amount of additional local control this raises the effective connectivity up to all-to-all (see SM~\cite{SM}).

\subsection*{Cyclic codes under global control}
The circular topology of the architecture introduced in this work is naturally matched to the implementation of cyclic stabilizer quantum error-correcting codes. In this Appendix we show that such codes are not only compatible with global control, but in fact their syndrome extraction can be implemented using only the globally driven routing.

\begin{figure}[t]
    \centering
    \includegraphics[width=0.7\columnwidth]{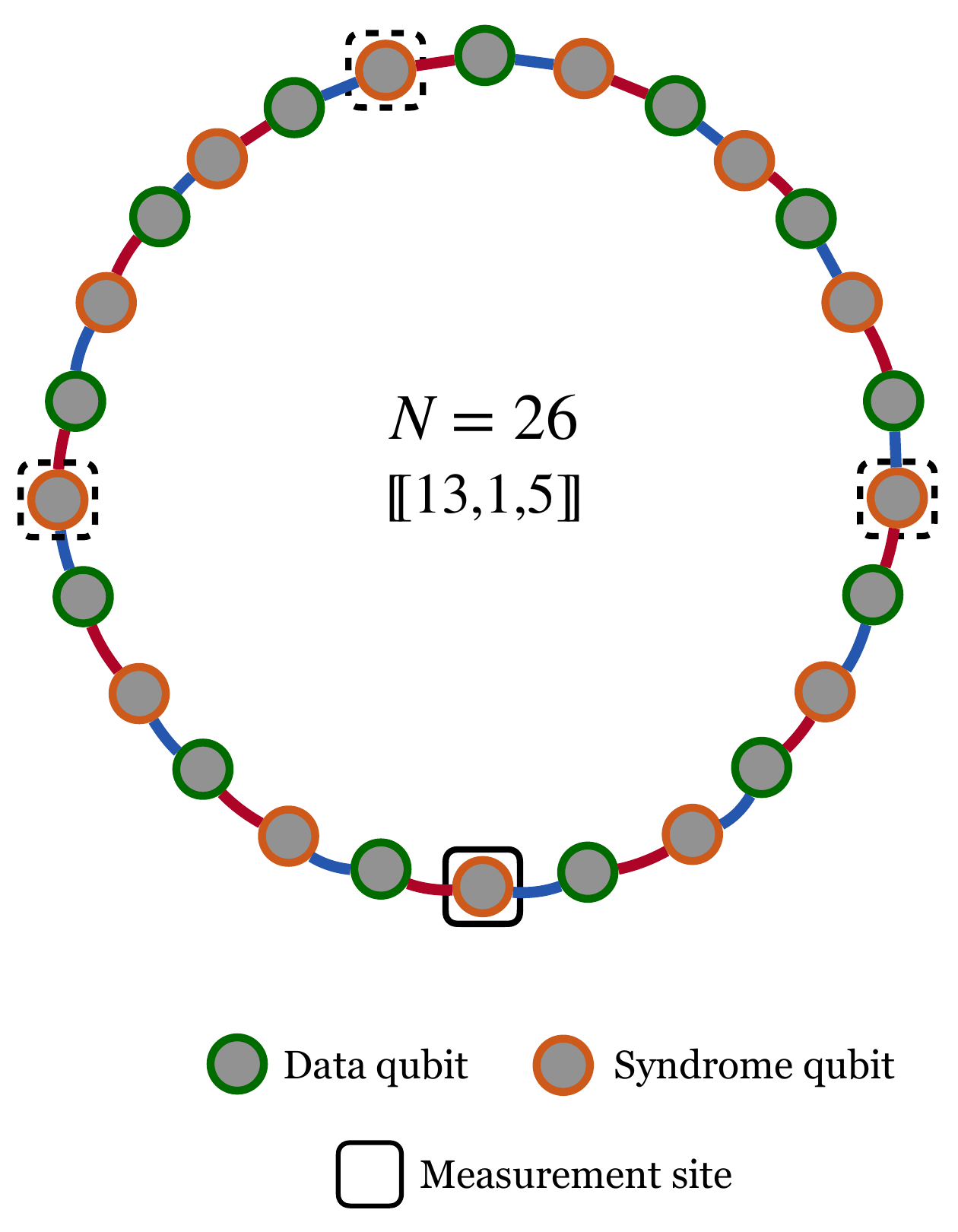}
    \caption{The cyclic $[\![13,1,5]\!]$ code laid out on our ring architecture with $N=2n=26$ total qubits, using $n=13$ data qubits (green outline) and syndrome qubits (orange outline), respectively, arranged in an alternating pattern. Red and blue bonds represent the two global controls $g(t)$ and $h(t)$. Squares around qubits denote local measurement sites, where syndrome qubits can be measured and reset.  In principle, only one site is required, shown with the solid square outline.  More sites can be evenly distributed to reduce the number of rounds of measurement $M$, exemplified with dashed square outlines for $\ell=4$ total measurement sites.  For the distance $d=5$ code with $n_s = 13$ syndromes, this results in $M=\lceil \frac{n_s}{\ell} \rceil = 4$ rounds of measurement.  As this code is of distance $d=5$, it can correct any two single-qubit errors.}
    \label{fig:codes}
\end{figure}

We consider a non-CSS cyclic stabilizer code $[\![n,1,d]\!]$ encoding one logical qubit into $n$ physical qubits with stabilizer group $\mathcal{S} = \langle g_0, \dots, g_{n-2} \rangle$, where the generators are obtained by cyclic translations of a base operator $g_0$, i.e.~$g_j = \mathrm{Cyc}^j(g_0)$ with $j=0,\dots,n-2$ and $\mathrm{Cyc}$ being the cyclic shift operator on the ring. 
A paradigmatic example is the family introduced in Ref.~\cite{Simakov2025}, where the base stabilizer, for various code distances $d$, takes the form
\begin{align}
&d=3, \, n=5: \quad \; \;  g_0 = Z X X Z I,  \nonumber \\
&d=5, \, n=13: \quad g_0 = ZIXXIZI^{\otimes 7} ,\nonumber \\
&d=7, \, n=21: \quad g_0 = ZIIXXXXIIZI^{\otimes 11},\nonumber \\
&d=9, \, n=29: \quad g_0=ZXXIIXIIXIIXXZI^{\otimes 15} \nonumber
\end{align}
and all other $n-2$ generators are obtained by cyclic permutation. The logical Paulis are transversal (global):
\begin{equation}
X_L = X^{\otimes n}, \quad Y_L = Y^{\otimes n}, \quad Z_L = Z^{\otimes n}.
\end{equation}

These codes can correct up to $t=\lfloor (d-1)/2 \rfloor$ errors, with the logical error rate scaling as $\epsilon_L \sim p^{(d+1)/2}$ under standard depolarizing noise models. The number of physical qubits scales linearly with code distance. In particular, one has $(d=3,\, n=5)$ corresponding to the five-qubit perfect code, $(d=5,\, n=13)$, $(d=7,\, n=21)$, and, more generally, $n = 4d - 7$ with $d \in 2\mathbb{N}+1$.

A key feature of cyclic codes is that they admit an implementation on a ring of alternating data and syndrome qubits, as shown in Fig.~\ref{fig:codes} for the explanatory $[\![13,1,5]\!]$ code. Specifically, we consider a register of $N=2n$ qubits arranged as
\begin{equation}
\cdots - S_i - D_i - S_{i+1} - D_{i+1} - \cdots
\end{equation}
where $D_i$ are data qubits and $S_i$ are syndrome qubits (ancillary qubits used for syndrome measurement)---such that $S_i \in {\rm E(O)}$ and $D_i \in {\rm O(E)}$.

Syndrome extraction is performed by entangling each syndrome qubit with a finite number of neighboring data qubits. Importantly, this can be realized using only nearest-neighbor iSWAP gates plus single-qubit gates~\cite{Simakov2022, Simakov2025}. The total number of physical qubits and distinct two-qubit couplings scale linearly with the distance: $N_{\mathrm{phys}} \sim N_{2q} \sim \mathcal{O}(d)$, in contrast to planar codes where $N_{\mathrm{phys}} \sim \mathcal{O}(d^2)$.

The crucial observation is that the stabilizer extraction circuit of cyclic codes comprises solely operations that act in parallel across the entire register, allowing for implementation with global gates alone. At each step, identical two-qubit gates are applied simultaneously to disjoint pairs of neighboring qubits. This exactly matches the global-control dynamics of our architecture, where the Hamiltonian alternates between $\hat{H}_{{\rm int},\mathcal{P}_{\rm E}}(t)$ and $\hat{H}_{{\rm int},\mathcal{P}_{\rm O}}(t)$
generating parallel iSWAP operations on alternating bonds. Single-qubit gates necessary for the syndrome extraction protocol are implemented in parallel on even- and odd-indexed qubits via~\eqref{eq:odd-even}. 

Due to the cyclic and translation-invariant structure of the dynamics, the syndrome qubits effectively ``rotate'' around the ring, sequentially interacting with the data qubits required to build up each stabilizer. At the end of the protocol, the stabilizer information is coherently accumulated into the syndrome qubits through a sequence of $\mathcal{O}(n)$ native two-qubit operations.
The only non-global operation required is the measurement of syndrome qubits. In our architecture, this is resolved by exploiting the intrinsic transport capability of the ring. After the stabilizer information has been encoded into all syndrome qubits, their quantum states are sequentially relocated to sites endowed with local read-out via the global transport dynamics, until all syndromes have been measured.

\clearpage
\onecolumngrid

\begin{center}
\vspace*{\baselineskip}
{\large \textbf{Supplemental Material}}
\end{center}

\setcounter{section}{0}
\setcounter{equation}{0}%
\setcounter{figure}{0}%
\setcounter{table}{0}%

\setcounter{page}{1}
\renewcommand{\thepage}{S\arabic{page}}
\renewcommand{\theequation}{S\arabic{equation}}
\renewcommand{\thefigure}{S\arabic{figure}}
\renewcommand{\thetable}{S\arabic{table}}


\section{iSWAP gate via XY interaction}
A single-qubit unitary rotation of angle $\theta$ about the $\mathbf{n} \in \mathbb{R}^3$ direction ($|\mathbf{n}|^2=1$) is defined as
\begin{eqnarray} 
R_\mathbf{n}(\theta)
     :=\exp[ -i ({\theta}/{2}) \mathbf{n} \cdot \bm{\sigma}]
     =\cos(\theta/2) I
    - i \sin(\theta/2) \mathbf{n} \cdot \bm{\sigma}\; .
\end{eqnarray} 
Based on this definition, let us enumerate the main single-qubit gates (up to a global phase):
\begin{itemize}
    \item $X = e^{i\pi/2} R_X(\pi)$, $Y = e^{i\pi/2} R_Y(\pi)$, $Z = e^{i\pi/2} R_Z(\pi)$;
    \item $H = e^{i\pi/2} R_\mathbf{n}(\pi)$ with $\mathbf{n} = \frac{\mathbf{x} + \mathbf{z}}{\sqrt{2}}$;
    \item $S = e^{i\pi/4} R_Z(\pi/2)$, $T = e^{i\pi/8} R_X(\pi/4)$.
\end{itemize}
Notice that the global phases are irrelevant since they do not affect measurement statistics.

\medskip

Let us set $\hbar=1$ for simplicity. Let us consider two nearest-neighbors qubits, $A$ and $B$, coupled via a XY interaction, whose Hamiltonian is 
\begin{align}
    \label{eq:H-tot-XX}
    \hat{H} &=  \hat{H}_A + \hat{H}_B + \hat{H}_{\rm int} \\ &:= \frac{\omega_A}{2}\hat{\sigma}^{(z)}_A + \frac{\omega_B}{2}\hat{\sigma}^{(z)}_B + \frac{g}{2} (\hat{\sigma}^{(x)}_A \otimes \hat{\sigma}^{(x)}_B + \hat{\sigma}^{(y)}_A \otimes \hat{\sigma}^{(y)}_B),
\end{align}
where $\hat{\sigma}_\chi^{(\alpha)}$ are the Pauli matrices with $\chi \in \{A,B\}$ and $\alpha \in \{x,y,z\}$.
Defining the raising and lowering operators $\hat{\sigma}^{(\pm)}_\chi = 1/2 (\hat{\sigma}^{(x)}_\chi \pm i\hat{\sigma}^{(y)}_\chi)$, the interaction term simply becomes
\begin{eqnarray}
    \label{eqsup:interaction-swap}
    \hat{H}_{\rm int} = g \left(\hat{\sigma}^{(+)}_A \otimes \hat{\sigma}^{(-)}_B + \hat{\sigma}^{(-)}_A \otimes \hat{\sigma}^{(+)}_B \right).
\end{eqnarray}
The effect of $\hat{H}_{\rm int}$ is to create an excitation on one qubit and annihilate the other. Let us write the whole Hamiltonian $\hat{H}$ as a matrix in the joint computational basis of the two qubits, $\{|00\rangle, |01\rangle, |10\rangle, |11\rangle\}_{AB}$, such that
\begin{equation}\label{eqsub:matrixH}
    [H] = \begin{pmatrix}
        \frac{\omega_A+\omega_B}{2} & 0 & 0 & 0 \\
        0 & \frac{\omega_A-\omega_B}{2} & g & 0 \\
        0 & g & \frac{-\omega_A+\omega_B}{2} & 0 \\
        0 & 0 & 0 & \frac{-\omega_A-\omega_B}{2}
    \end{pmatrix} \ .
\end{equation}
Note that $|00\rangle$ and $|11\rangle$ are eigenstates while $\{|01\rangle, |10\rangle\}$ evolves non-trivially. The sub-matrix 2-by-2 acting on $\{|01\rangle, |10\rangle\}$, 
\begin{equation}
    [H_{\rm sub}] =  \begin{pmatrix}
        \Delta/2 & g \\
        g & -\Delta/2  \\
    \end{pmatrix} ,
\end{equation}
has eigenvalues $E_{\pm} = \pm  \sqrt{(\Delta/2)^2 +g^2}$, where $\Delta := \omega_A -\omega_B$ is the detuning between the two qubits. Being $\hat{H}_0$ the non-interacting part of Eq.~\eqref{eq:H-tot-XX}, the unitary transformation to go in the interaction picture, i.e., the reference frame rotating at the frequencies of the qubits, is given by $\hat{U}_{\rm rf}(t) := e^{-it \hat{H}_0} = e^{-i\omega_A t/2 \hat{\sigma}_A^{(z)}} \otimes e^{-i\omega_B t/2 \hat{\sigma}_B^{(z)}}$. Thus, the Hamiltonian $\hat{H}$ reduces to 
\begin{eqnarray}
    \hat{H}_{\rm rf}(t) &=& \hat{U}_{\rm rf}(t) \hat{H} \hat{U}_{\rm rf}^\dagger(t) \\
    &=& g \left(e^{i \Delta t} \hat{\sigma}^{(+)}_A \otimes \hat{\sigma}^{(-)}_B + e^{-i \Delta t}\hat{\sigma}^{(-)}_A \otimes \hat{\sigma}^{(+)}_B \right) .
\end{eqnarray}

\subsection{Resonant case}
In the resonant scenario, $\Delta = 0$, i.e., $\omega_A=\omega_B$, the effective Hamiltonian reduces to the time-independent interacting Hamiltonian~\eqref{eqsup:interaction-swap}: $\hat{H}_{\rm rf}(t)|_{\Delta=0} = \hat{H}_{\rm int}$. In particular, the unitary dynamics is given by 
\begin{eqnarray}
    \hat{U}_{\rm dyn}(t) = \exp\left(-i t \hat{H}_{\rm int}\right) = \begin{pmatrix}
        1 & 0 & 0 & 0 \\
        0 & \cos(gt) & -i\sin(gt) & 0 \\
        0 & -i\sin(gt) & \cos(gt) & 0 \\
        0 & 0 & 0 & 1
    \end{pmatrix},
\end{eqnarray}
that by setting $gt=3\pi/2$ reduces to an iSWAP gate between the two qubits, i.e.
\begin{eqnarray}
    \label{eq:iswap}
    \hat{U}_{\rm dyn}\left(\frac{3\pi}{2g}\right) =  \begin{pmatrix}
        1 & 0 & 0 & 0 \\
        0 & 0 & i & 0 \\
        0 & i & 0 & 0 \\
        0 & 0 & 0 & 1
    \end{pmatrix} := \text{iSWAP} \ .
\end{eqnarray}
Notice that, if one sets $gt=\pi/2$, then the unitary evolution reduces to an iSWAP${}^\dagger$, i.e., an iSWAP except for a relative minus sign in the central submatrix; we call it $\text{ iSWAP}^\dagger = \hat{U}_{\rm dyn}(\frac{\pi}{2g})$. To pass from one operation to the other is sufficient to apply local phase corrections:
\begin{equation}
    (Z\otimes I)\, \hat{U}_{\rm dyn}\!\left(\tfrac{\pi}{2g}\right)\, (Z\otimes I) = \text{iSWAP},
\end{equation}
where $Z$ and $I$ denote the Pauli-$Z$ and identity operators acting on qubits $A$ and $B$, respectively. The single-qubit $Z$-operation can be induced by a pulse rotation $ Z \simeq R_Z(\pi)$.

\subsection{Off-resonant case}
When the two interacting qubits are far off resonance, $|\Delta|/g \gg 1$, the exchange terms proportional to $e^{\pm i \Delta t}$ oscillate very rapidly, so that within the rotating-wave approximation (RWA) their first-order contribution averages to zero. We now investigate what happens at the level of second-order perturbation theory. Treating the SWAP interaction~\eqref{eqsup:interaction-swap} as the perturbative term, we note that its action is nonzero only on the off-diagonal matrix elements:
\begin{equation}
    \langle 01|\hat{H}_{\rm int}|10\rangle = \langle 10|\hat{H}_{\rm int}|01\rangle =  g.
\end{equation}
We then apply the Schrieffer--Wolff (SW) transformation, i.e., the unitary generated by $\hat{\cal W}$ that removes the off-diagonal terms of $\hat{H}_{\rm int}$ to first order, such that
\begin{equation}
    [\hat{H}_0,\hat{\cal W}] = -\hat{H}_{\rm int}.
\end{equation}
In our case, the only nonzero matrix elements of $\hat{\cal W}$ are
\begin{equation}
    \langle 01|\hat{\cal W}|10\rangle = \frac{\langle 01|\hat{H}_{\rm int}|10\rangle}{E_{01}-E_{10}} = \frac{g}{\Delta}, 
    \qquad 
    \langle 10|\hat{\cal W}|01\rangle = -\frac{g}{\Delta},
\end{equation}
so that $\hat{\cal W}$ is a $2\times 2$ block matrix,
\begin{equation}
    [\mathcal{W}] = \begin{pmatrix}
        0 & g/\Delta \\
        -g/\Delta & 0 \\
    \end{pmatrix},
\end{equation}
and $\hat{\cal W} = - \hat{\cal W}^\dagger$. The SW effective Hamiltonian to second order is
\begin{equation}
    \hat{H}_{\rm eff}^{(2)} = \hat{H}_0 + \frac{1}{2}[\hat{\cal W},\hat{H}_{\rm int}].
\end{equation}
Computing the diagonal matrix elements of the second-order term gives
\begin{equation}
    \langle 01|\tfrac{1}{2}[\hat{\cal W},\hat{H}_{\rm int}]|01\rangle = -\langle 10|\tfrac{1}{2}[\hat{\cal W},\hat{H}_{\rm int}]|10\rangle = \frac{g^2}{\Delta}.
\end{equation}
We thus arrive at the perturbative contribution
\begin{equation}
    \delta \hat{H}^{(2)} = \frac{g^2}{2\Delta} \left( \hat{\sigma}_A^{(z)} - \hat{\sigma}_B^{(z)} \right).
\end{equation}
The above second-order perturbation corresponds to a correction to the qubit frequencies, and consequently to faster single-qubit dynamics.

\subsection{Correcting the iSWAP into the pure SWAP gate}

Let us consider the two qubits, labeled as $A$ and $B$. Above, we have shown that when the two qubits are resonant in frequency ($\Delta=0$), the XY interaction does not give rise to an exact SWAP gate, but rather to an iSWAP acting on the Hilbert space of two qubits whose computational basis is $\{|00\rangle, |01\rangle, |10\rangle, |11\rangle\}_{AB}$. 
In the dynamics of our two-qubit system, the relative phase between $|01\rangle$ and $|10\rangle$ has to be corrected in order to induce a perfect exchange operation. 

In order to do so, we notice that the iSWAP can be turned into a conventional SWAP via a CZ gate and local rotations. In particular we have that~\cite{krizan2025} 

\begin{equation}\label{supp:iswap_to_swap_local}
    \text{SWAP} = (S^\dagger \otimes S^\dagger) \ \text{iSWAP} \ \text{CZ}  \ ,
\end{equation}

\begin{center}
\begin{quantikz}
\lstick{} & \gate[wires=2]{\mathrm{SWAP}} & \qw \\
\lstick{} &                          & \qw
\end{quantikz}
=
\begin{quantikz}
\lstick{} 
  & \ctrl{1} 
  & \gate[wires=2]{\mathrm{iSWAP}} 
  & \gate{S^\dagger} 
  & \qw \\
\lstick{} 
  & \ctrl{-1} 
  & 
  & \gate{S^\dagger} 
  & \qw
\end{quantikz}
\end{center}
with $[\text{iSWAP}, \text{CZ}] = 0$. Such a protocol can be implemented in superconducting qubits by combining iSWAP and CZ operations via tunable couplers, as shown in Ref.~\cite{Sung2021}.
Above we have that 

\begin{equation}
    \text{SWAP} = \begin{pmatrix}
        1 & 0 & 0 & 0 \\
        0 & 0 & 1 & 0 \\
        0 & 1 & 0 & 0 \\
        0 & 0 & 0 & 1
    \end{pmatrix}, \qquad \text{CZ} = \begin{pmatrix}
        1 & 0 & 0 & 0 \\
        0 & 1 & 0 & 0 \\
        0 & 0 & 1 & 0 \\
        0 & 0 & 0 & -1
    \end{pmatrix},
\end{equation}
and $\text{iSWAP}$ as in Eq.~\eqref{eq:iswap}. Notice that $S^\dagger \simeq R_Z(-\pi/2)$. 
In the framework of our work, since we have control on XY interaction and consequently on native iSWAP gate single local operations on the qubits, we would write Eq.~\eqref{supp:iswap_to_swap_local} only in terms of iSWAP and single-qubit operations. This can be done combining three iSWAP and local $\sqrt{X} \simeq R_X(\pi/2)$~\cite{Schuch2003, Marvian2024}, i.e.,

\begin{equation}\label{supp:iswap_to_swap_local2}
    \text{SWAP} =  (I \otimes \sqrt{X}) \ \text{iSWAP} \ (\sqrt{X} \otimes I) \ \text{iSWAP} \  (I \otimes \sqrt{X}) \  \text{iSWAP},
\end{equation}

\begin{center}
\begin{quantikz}
\lstick{} & \gate[wires=2]{\mathrm{SWAP}} & \qw \\
\lstick{} &                          & \qw
\end{quantikz}
=
\begin{quantikz}
\lstick{} 
  & \gate[wires=2]{\mathrm{iSWAP}} & \qw
  & \gate[wires=2]{\mathrm{iSWAP}} & \gate{\sqrt{X}}
  & \gate[wires=2]{\mathrm{iSWAP}} & \qw & \\
\lstick{} 
  &                              & \gate{\sqrt{X}}
  &                              & \qw
  &                              & \gate{\sqrt{X}}
  &                              
\end{quantikz}
\end{center}

We have then corrected the native iSWAP into a pure SWAP gate, which will be essential for the dynamics of our quantum processor.
Incidentally, from the decomposition above we also see that the CZ gate can be decomposed in terms of iSWAP as follows:

\begin{center}
\begin{quantikz}
\lstick{} 
  & \gate[wires=2]{\mathrm{iSWAP}} & \qw
  & \gate[wires=2]{\mathrm{iSWAP}} & \gate{\sqrt{X}}
  & \gate[wires=2]{\mathrm{iSWAP}} &  & \gate{S} \qw &  \gate[wires=2]{\mathrm{iSWAP}^\dagger} & \\
\lstick{} 
  &                              & \gate{\sqrt{X}}
  &                              & \qw
  &                              & \gate{\sqrt{X}}
  &         \gate{S}        \qw       &    &    
\end{quantikz}
= \begin{quantikz}
\lstick{} 
  & \ctrl{1} 
 \qw & \\
\lstick{} 
  & \ctrl{-1} 
  & \qw
\end{quantikz}
\end{center}
As a result, this also shows how to perform a CZ using the keyhole architecture described below.

\section{Physical implementation of the architecture with superconducting qubits}

We now present a physical realization of the aforementioned scheme using superconducting qubits. In modern superconducting hardware, the ZZ interaction is typically an undesired parasitic effect~\cite{Sung2021, Zhao2021}. This motivates the Hamiltonian model adopted here, where the ZZ term is assumed to be eliminated.

Since it is experimentally challenging---though not impossible~\cite{Kounalakis_2018, krizan2025}---to switch qubit–qubit interactions fully on and off, an alternative approach is to keep the XY interactions always on, and instead suppress them effectively by globally detuning selected pairs of superconducting qubits using global flux lines~\footnote{A flux line allows one to tune the internal frequencies of superconducting qubits. In particular, for flux-tunable qubits (such as SQUID-based transmons or flux qubits), the qubit frequency is given by $\omega_{\rm q}(t) = \omega_{\rm q}(\Phi(t)) = \sqrt{8E_{\rm C}E_{\rm J}(\Phi)}$, where $E_{\rm C}$ is the charging energy, determined by the qubit capacitance, $\Phi$ is the magnetic flux, and $E_{\rm J}(\Phi) = E_{\rm J}^{\rm max} \cos(\pi \Phi/\Phi_0)$ is the flux-dependent Josephson energy, with $\Phi_0$ the flux quantum.}. This realizes the same logical control structure as in the interaction-modulated scheme previously described.
For clarity of presentation, let us consider four types of qubits, $\mathcal{S}:=\{A, B, C, D\}$. All qubits of the same type share a global flux line and therefore have a globally-controlled internal frequency~\footnote{We note that the architecture requires at most four global control lines (one per qubit species), but this number can be reduced to three by introducing types-dependent static frequency offsets.}. Following the structure of the previously introduced model, we arrange the qubits in a ring of $N$ sites with the repeating pattern $ABCDABCD\ldots$. Each qubit interacts with its nearest neighbor through an always-on XY coupling. A pictorial illustration of the architecture for $N = 8$ computational qubits is shown in Fig.~\ref{fig:ring-freq}.
As discussed earlier, at each step we must selectively activate only the interactions involving the pairs the subsets $\mathcal{P}_{\rm E}$ or $\mathcal{P}_{\rm O}$. In the present formulation, these correspond respectively to the pairs $(AB, CD)$ and $(BC, DA)$. When two qubits interact via an XY coupling with strength $g$ but are brought far out of resonance, the effective interaction is strongly suppressed, up to a residual second-order local frequency correction depending on the detuning $\Delta$. Thus, by globally shifting the frequencies of qubit types $A$, $B$, $C$, and $D$ using four global flux lines, we can engineer time windows in which only the desired qubit pairs are resonant and therefore interact.
In this way, the same dynamical pattern of iSWAP operations used in the interaction-controlled version of the scheme is reproduced here through frequency control. 

\begin{figure}[t]
    \centering
    \includegraphics[width=0.6\linewidth]{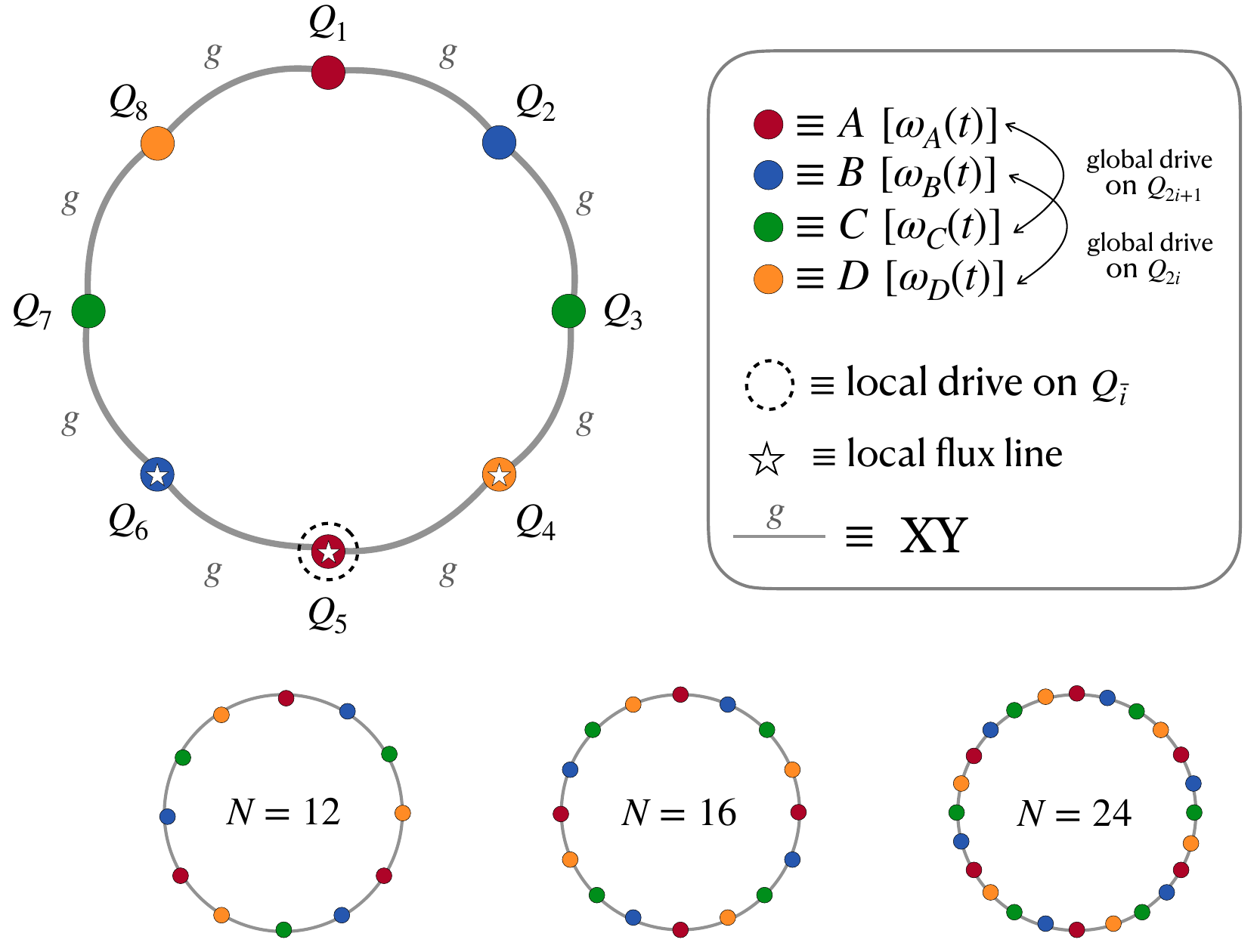}
    \caption{Superconducting quantum computing scheme with global control on internal qubit frequencies and always-on XY interactions for $N=8$ computational qubits. The legend is self-explanatory. By scaling $N$, the number of controls remains constant.}
    \label{fig:ring-freq}
\end{figure}

The Hamiltonian of the system is given by $\hat{H}(t) := \hat{H}_0(t) + \hat{H}_{\rm int} + \hat{H}_{{\rm drive}}(t)$ where 
\begin{eqnarray}
    &&\displaystyle \hat{H}_0(t):=\sum_{\chi \in \mathcal{S}} \sum_{i \in \chi} \frac{\omega_i(t)}{2} \hat{\sigma}_i^{(z)} \ , \label{eq:H0_freq} \\
    &&\displaystyle \hat{H}_{\rm int}:= \sum_{\langle i,j\rangle} \dfrac{g}{2} \left(\hat{\sigma}^{(x)}_{i} \otimes \hat{\sigma}^{(x)}_{j} + \hat{\sigma}^{(y)}_{i} \otimes \hat{\sigma}^{(y)}_{j}\right) \ \label{eq:Hint_freq} \ ,
\end{eqnarray}
and $\hat{H}_{\rm drive}(t)$ as in the main text. Now, in addition to the local drive term, the global symmetry-breaking term is contained in $\hat{H}_0$ for $i= \{\bar{i}-1, \bar{i}, \bar{i} +1 \}$.

Let $\Delta_{\chi\chi'}:= |\omega_{\chi}(t) - \omega_{\chi'}(t)|$ be the detuning between the two distinct types $\chi\neq \chi' \in \mathcal{S}.$ Thus, the dynamical, single-qubit gate, and two-qubit gate protocols presented above are formulated as follow:
\begin{eqnarray}
    &&\mathcal{T}_{AB\&CD} : \ \Delta_{BC/DA}(t)= \Delta_{\rm off} \  \wedge \ \Delta_{AB/CD}(t)= 0  \ ,  \label{eq:dyn-abcd} \nonumber \\
    &&\mathcal{T}_{BC\&DA} : \ \Delta_{AB/CD}(t)= \Delta_{\rm off} \  \wedge \ \Delta_{BC/DA}(t)= 0  \ , \nonumber  \label{eq:dyn-bcda}
\end{eqnarray}
that correspond to Eqs.~(3) and (4) of the main text, respectively, and 
\begin{eqnarray}
    &&\mathcal{T}_{\rm 1q}: \Delta_{\chi,\chi'} = \Delta_{\rm off} \ , \ \ \forall \ \chi\neq\chi' \in \mathcal{S} \ \wedge \ \Omega(t), \phi(t) \neq 0 \ , \nonumber \\
    &&\mathcal{T}_{\rm 2q}: \Delta_{\chi,\chi'} = \Delta_{\rm off} \ , \ \ \forall \ \chi\neq\chi' \in \mathcal{S} \  \bigg/  \nonumber  \begin{cases}
         \Delta_{\bar{i}, \bar{i}+1}(t) = \Delta_{\rm off},\quad \Delta_{\bar{i}-1, \bar{i}}(t)=0 \ , \\
         \Delta_{\bar{i}, \bar{i}+1}(t) = 0,\quad \Delta_{\bar{i}-1, \bar{i}}(t)=\Delta_{\rm off} \ . \nonumber
     \end{cases}
\end{eqnarray}
Above, $\Delta_{\rm off}$ denotes the off-resonant detuning. By combining the global dynamics with these protocols, we obtain a realizable globally-controlled quantum computer for superconducting qubits. However, we remark that the initial scheme -- although more experimentally demanding -- could still be implemented in practice with superconducting qubits. Indeed, as demonstrated in Ref.~\cite{Kounalakis_2018}, when two qubits are connected through a tunable coupler (i.e., an ancillary qubit), adjusting the coupler flux allows one to reach operating points where the swap interaction is effectively turned off. In this case, an $N$-qubit processor would require a total of $2N$ qubits, including the ancillary couplers~\cite{Sung2021}.

\section{Keyhole geometry}

\begin{figure}[t]
    \centering
    \includegraphics[width=1.0\linewidth]{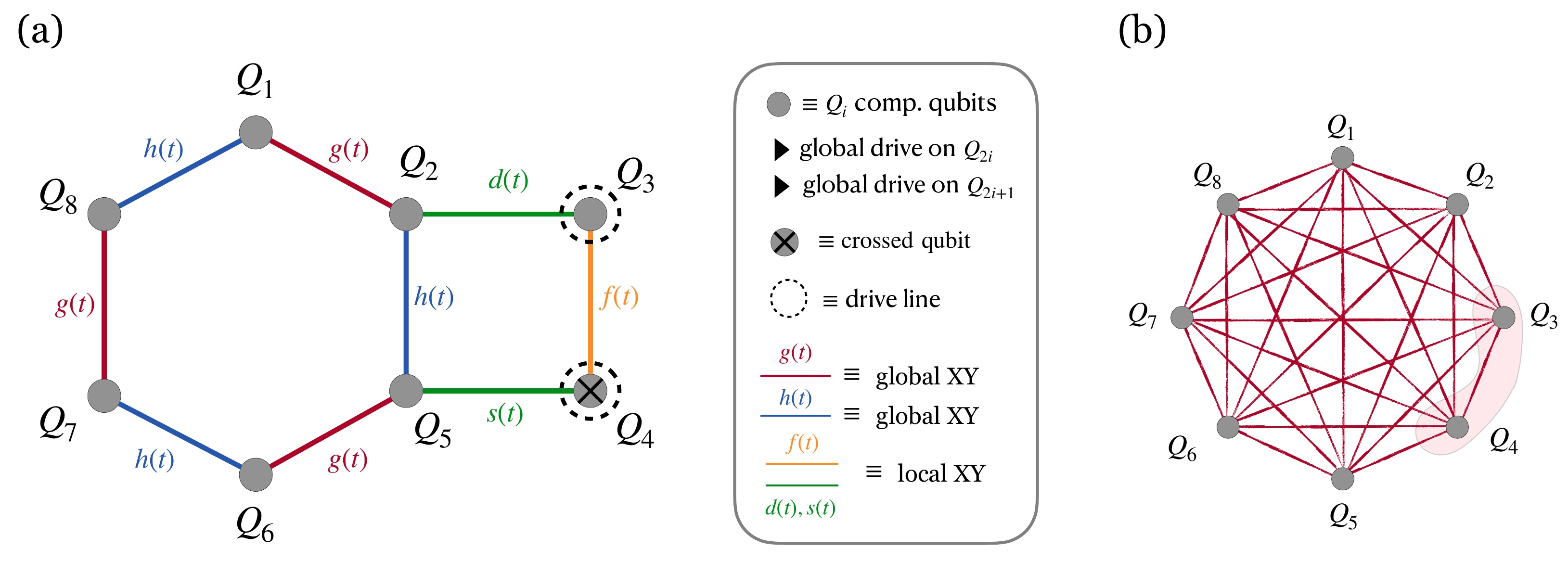}
    \caption{(a) Globally-controlled keyhole quantum computer for $N=8$ computational qubits. (b) Connectivity graphs corresponding to (a). Here, connectivity refers to all possible pairs that can be brought into the locally controlled region solely via swap dynamics, excluding local computational swaps; note that by occasionally acting on the local control region, effective all-to-all connectivity can in principle be reached even for the ring geometry.}
    \label{fig:keyhole}
\end{figure}

By deforming the closed-ring geometry into a {\it keyhole} geometry---consisting of a ring that stores the computational qubits and a short corridor that acts as an actively and locally controlled region of the processor (see Fig.~\ref{fig:keyhole}(a))---we obtain a layout suitable for implementing universal global-control schemes with full effective all-to-all connectivity, as shown in Fig.~\ref{fig:keyhole}(b). Logical operations are performed in the corridor: single-qubit rotations are applied to the two qubits at its ends through a shared drive line acting on both of them (one of the two being a crossed qubit~\cite{menta2025building, menta2026actuators}, i.e., a qubit with an enhanced Rabi frequency compared to its neighbor). 
Two-qubit iSWAP gates can also be implemented between the qubits located on the corridor’s sites thanks to the locally controlled XY interaction $f(t)$ shown in Fig.~\ref{fig:keyhole}(a). In particular, by turning off all global interactions as well as the corridor-link interactions $s(t)$ and $d(t)$, and activating only $f(t)$, one can implement the native two-qubit gate.
The locally controlled links $s(t)$ and $d(t)$ are instead used to move, whenever needed, any chosen pair of qubits from the memory ring into the two end sites of the corridor by means of SWAP operations. This relocation is enabled by the global rotational dynamics acting on the memory ring, which transports quantum information around the ring until the desired qubits reach the appropriate positions.
We stress that such a keyhole geometry can also be realized using always-on XY interactions, together with global and local control of the qubit internal frequencies, as discussed above.
We finally notice that these geometries can be easily deformed into different graph-theoretic variants that minimize hardware resources while preserving full connectivity. For instance, one may imagine designing a lollipop graph geometry, where the corridor consists of a single locally controlled link connecting a qubit on the ring to a locally driven ancillary qubit.

\section{Quantum circuit compilation}

We formalize here the routing problem for the architectures of the main text and show
constructively that the globally controlled dynamics is sufficient to bring any computational
qubit -- and, in the keyhole geometry, any \textit{pair} of qubits -- to the sites where the
required gates are available. Throughout, $N$ denotes the (even) number of ring sites, as in the
main text.

\subsection{Routing on the ring}

Label the ring sites by $0,\dots,N-1$, all indices understood modulo $N$, and split them into the
two sublattices ${\rm E}=\{0,2,\dots,N-2\}$ and ${\rm O}=\{1,3,\dots,N-1\}$, each of size $N/2$.
Local single-qubit gates are available only on a nonempty subset
$\mathcal{S}_{1{\rm QG}}\subseteq\{0,\dots,N-1\}$ of sites---a single site $\bar{i}$ in the
minimal model of Fig.~\ref{fig:ring-int} of the main text---so any logical single-qubit gate must be preceded by
transport of the target state to a site of $\mathcal{S}_{1{\rm QG}}$.

The two drive windows $\mathcal{T}_g$ and $\mathcal{T}_h$ of Eqs.~\eqref{eq:alternative-g-h} and
\eqref{eq:alternative-h-g} of the main text, once each iSWAP has been corrected into a SWAP, act on the qubit
positions as the two alternating nearest-neighbor SWAP layers
\begin{equation}
\pi_g = (0\,1)(2\,3)\cdots(N-2\;N-1),
\qquad
\pi_h = (1\,2)(3\,4)\cdots(N-1\;0),
\end{equation}
where $\pi(i)$ denotes the site to which the state stored at $i$ is moved. Their products
\begin{equation}
\pi_{\rm L} := \pi_g\circ\pi_h,
\qquad
\pi_{\rm R} := \pi_h\circ\pi_g = \pi_{\rm L}^{-1},
\end{equation}
with $\circ$ the usual composition (rightmost factor applied first), are the elementary transport
primitives. Explicitly,
\begin{equation}
\label{eq:LR_powers}
\pi_{\rm L}^{\,k}(e) \equiv e-2k,
\qquad
\pi_{\rm L}^{\,k}(o) \equiv o+2k
\qquad (\text{mod } N),
\end{equation}
for $e\in{\rm E}$ and $o\in{\rm O}$, and $\pi_{\rm R}^{\,k}$ with both signs reversed: the two
sublattices are rigidly translated in opposite directions, by one sublattice step per application.

\begin{theorem}[Single-qubit routing completeness]
\label{thm:1Q_routing}
For every pair of sites $p_0,t\in\{0,\dots,N-1\}$ there exists a finite word in $\pi_g$ and
$\pi_h$ mapping $p_0$ to $t$. In particular, any computational qubit can be brought to any site
at which local control is available.
\end{theorem}

\begin{proof}
By Eq.~\eqref{eq:LR_powers}, $\pi_{\rm L}$ generates the full cyclic translation group of each
sublattice, so any two sites of equal parity are connected by $\pi_{\rm L}^{\,k}$ for some
$0\le k<N/2$. If $p_0$ and $t$ have opposite parity it suffices to apply $\pi_g$ (or $\pi_h$)
once beforehand, since every transposition in their cycle decomposition joins an even to an odd
site, and then to proceed as above.
\end{proof}

The proof is constructive and fixes the compilation cost. Once the parities match, let
$d:=\tfrac{1}{2}(t-p) \ ({\rm mod} \ N/2)$; the minimal number of composite primitives is
$\min\{d,\,N/2-d\}$, its sign selecting $\pi_{\rm L}$ or $\pi_{\rm R}$, with one additional layer
$\pi_g$ or $\pi_h$ when the parities differ. Since each layer is one global SWAP, i.e.\ three
global iSWAPs and three global single-qubit pulses by Eq.~\eqref{supp:iswap_to_swap_local2}, the
worst-case transport cost is $\mathcal{O}(N)$ global pulses, independent of which qubit is
addressed.

If $\Pi$ is the permutation realized by a routing word that sends the qubit of interest to
$t\in\mathcal{S}_{1{\rm QG}}$, a physical gate $U_t$ applied there implements the logical operation
\begin{equation}
U_{\rm logical} = \Pi^{\dagger}
\left(\mathbb{I}^{\otimes t}\otimes U_t\otimes\mathbb{I}^{\otimes(N-t-1)}\right)\Pi \ ,
\end{equation}
although in practice the routing need not be undone immediately: the computation may simply
proceed in the permuted frame.

As an illustration, take $N=8$ with $\mathcal{S}_{1{\rm QG}}=\{4\}$ and the state to be addressed
initially at $p_0=3$. The parities differ, so one applies $\pi_g$, which moves it to site $2$; a
single $\pi_{\rm R}$ then completes the transport, since $\pi_h(\pi_g(2))=4$.

\subsection{Pair routing and universality in the keyhole geometry}

Bringing a \textit{chosen pair} of qubits simultaneously into a locally controlled link is a
strictly stronger requirement, because the global layers translate the two sublattices rigidly
and hence shift the separation of any even--odd pair by a fixed amount at every step. In the
keyhole geometry of Fig.~\ref{fig:keyhole} this obstruction is lifted by the two ancillary ports
$a$ and $b$, attached to the ring sites $0$ and $1$ through the locally controlled links, whose
swaps
\begin{equation}
\pi_a = (0\,a), \qquad \pi_b = (1\,b)
\end{equation}
commute with the global ring motion, the ports being inert under $\pi_g$ and $\pi_h$.

\begin{theorem}[Pair routing in the keyhole]
\label{thm:2Q_routing}
For any two distinct ring sites $p\neq q$ there exists a finite sequence in
$\langle \pi_g,\pi_h,\pi_a,\pi_b\rangle$ transferring the states initially stored at $p$ and $q$
to the ports $a$ and $b$, respectively.
\end{theorem}

\begin{proof}
By Theorem~\ref{thm:1Q_routing}, route the state at $p$ to site $0$ and apply $\pi_a$, parking it
at $a$. The second state has meanwhile been displaced to some site $q'$; route it to site $1$,
again by Theorem~\ref{thm:1Q_routing}, which leaves $a$ untouched, and apply $\pi_b$. The
composite word is $\Pi^{(p,q\to a,b)} = \pi_b\circ\Pi^{(q'\to 1)}\circ\pi_a\circ\Pi^{(p\to 0)}$.
\end{proof}

The corridor supports arbitrary single-qubit rotations on $a$ and $b$ together with a native
iSWAP between them, so any two-qubit unitary can be synthesized on $(a,b)$. Combined with
Theorem~\ref{thm:2Q_routing}, and with the inverse word returning the transformed states to their
original sites, this renders the effective two-qubit connectivity of the keyhole all-to-all and
the architecture universal: every logical circuit compiles into global routing on the ring, local
operations in the corridor, and inverse routing, with all locality constraints absorbed into the
routing layer.

The same conclusion holds for the plain ring, but only once the two locally controlled links at
$\bar{i}$ are also used as transport primitives. Under the global layers alone, the pairs that
can be brought into the locally controlled region form a single cycle through all $N$ qubits:
connected, and therefore still sufficient for universality by composition along it, but not
all-to-all.

\end{document}